\documentclass[12 pt]{extarticle}
\usepackage{amsmath}
\usepackage{amssymb}
\usepackage[a4paper, total={6in, 9in}]{geometry}
\usepackage{amsthm,url,tikz}
\usepackage[utf8]{inputenc}
\usepackage[english]{babel}
\usepackage{graphicx}
\usepackage{multirow} 
\usepackage{makecell}
\usepackage{hyperref}
\usepackage{adjustbox}
\numberwithin{equation}{section}
\usepackage{makecell}
\usepackage{bm}
\usepackage{xcolor}
\newtheorem{theorem}{Theorem}[section]
\newtheorem{proposition}[theorem]{Proposition}
\newtheorem{corollary}[theorem]{Corollary}
\newtheorem{lemma}[theorem]{Lemma}
\newtheorem{definition}[theorem]{Definition}

\newtheorem{example}[theorem]{Example}
\newtheorem{remark}[theorem]{Remark}
\usepackage{enumerate}
\usepackage{array}
\usepackage[nottoc,notlot,notlof]{tocbibind} 

\title{ Function-Correcting $b$-symbol Codes for Locally ($\lambda,\rho, b$)-Functions }
\author{Gyanendra K. Verma,
Anamika Singh, and Abhay Kumar Singh
 \footnote{ 
 Gyanendra K. Verma is financially supported by UAEU-AUA grant G00004614. Gyanendra K. Verma is with the Department of Mathematical Sciences,
UAE University, PO Box 15551, Al Ain, UAE email: gkvermaiitdmaths@gmail.com. Anamika Singh and Abhay Kuamr Singh are with the Department of Mathematics and Computing, Indian Institute of Technology (ISM) Dhanbad, Dhanbad,
 India email: anamikabhu2103@gmail.com, email: abhay@iitism.ac.in.
 Corresponding author: Abhay Kumar Singh.}}

\date{}

\begin{document}
	\maketitle
\begin{abstract} 

The family of functions plays a central role in the design and effectiveness of function-correcting codes. By focusing on a well-defined family of functions, function-correcting codes can be constructed with minimal length while still ensuring full error detection and correction within that family. In this work, we explore the concept of locally $(\lambda,\rho)$-functions for $b$-symbol read channels and investigate the optimal redundancy of the corresponding function-correcting $b$-symbol codes (FCBSC) by introducing the notions of locally  $(\lambda,\rho,b)$-functions. First, we discuss the values of $\lambda$ and $\rho$ for which a function can be considered as a locally $(\lambda,\rho)$-function in $b$-symbol metric. The findings improve some known results in the Hamming metric and present several new results in the $b$-symbol metric. Then we investigate the optimal redundancy of $(f,t)$-FCBSCs for locally $(\lambda,\rho,b)$-functions. We establish a recurrence relation between the optimal redundancy of $(f,t)$-function-correcting codes for the $(b+1)$-symbol read and $b$-symbol read channels. We present an upper bound on the optimal redundancy of $(f,t)$-function-correcting $b$-symbol codes for general locally ($\lambda,\rho$, $b$)-functions by associating it to the minimum achievable length of $b$-symbol error-correcting codes and traditional Hamming-metric codes, given a fixed number of codewords and a specified minimum distance.  We derive some explicit upper bounds on the redundancy of $(f,t)$-function-correcting $b$-symbol codes for locally $(\lambda,2t,b)$-functions. Moreover, for the case where $b=1$, we show that a locally ($3,2t,1$)-function achieves the optimal redundancy of $3t$. Additionally, we explicitly investigate the locality and optimal redundancy of FCBSCs for the 
 $b$-symbol weight function and weight distribution function for $b\geq1$.

\end{abstract}
\textbf{Keywords}: Function-correcting codes, error-correcting codes, $b$-symbol read channels, redundancy.\\ 
	\textbf{Mathematics subject classification}: 94B60, 94B65.\\

\section{Introduction}
In classical communication channels, messages are transmitted through error-prone channels. Error-correcting codes are designed to recover the original message despite these errors, ensuring that the entire message can be accurately reconstructed. In many practical scenarios, the receiver is more interested in a specific function of the message rather than the entire message itself. For example, in sensor networks, multiple devices collect data like temperature or humidity, but the receiver typically only needs to compute summaries such as the average or maximum value. Similarly, in distributed computing systems like MapReduce, large datasets are processed across different servers, and only aggregate results like counts or totals are required. In blockchain and cryptographic protocols, verifying functions like hashes or digital signatures is essential, while access to the full data is unnecessary. In video or audio streaming, the receiver may only need information about stream quality or keyframes, rather than the entire media content. Network monitoring systems often focus on metrics like error rates or traffic patterns instead of inspecting all transmitted packets. Likewise, in remote control systems and IoT applications, devices commonly require simple commands or alerts rather than complete configuration files. These examples highlight how function-based communication is more efficient and relevant in many real-world applications.

In function-correcting codes (FCCs), both the communication channel and the chosen function play a central role. They determine the specific aspect or computation of the message that must be preserved during transmission, as well as how errors affect this process and how redundancy should be added to ensure accurate recovery. Unlike classical error-correcting codes, which aim to recover the entire message despite the presence of errors, FCCs are designed to ensure the accuracy of a particular function applied to the message, even if parts of the message are corrupted. This change in focus places the function at the core of both the encoding and decoding processes. Moreover, FCCs open new directions in coding theory by enabling the design of codes tailored to the structure and significance of specific functions, rather than treating all message components uniformly. As a result, FCCs offer a powerful and adaptable approach with strong potential for both theoretical advancement and practical application.

Function-correcting codes, as introduced in the breakthrough work by Lenz et al. \cite{Lenz2023}, offer a significant advantage in minimizing redundancy. Therefore, it is important to explore their performance across various types of communication channels using appropriate functions. The $b$-symbol read channel is a key candidate for the design of function-correcting codes, owing to its critical role in high-density data storage systems. $b$-symbol read channels have been developed, aiming to provide robust protection against errors occurring within symbol-pairs and $b$-symbols. The concept of symbol-pair read channels was first introduced by Cassuto and Blaum in \cite{symbol_pair}, and later extended to $b$-symbol read channels by Yaakobi et al. in \cite{b_symbol}. In a series of subsequent works \cite{Chen2022,Chen2025,Ding2018,Elishco2020,Luo2024,Song2018,Verma2025}, various bounds and coding-theoretic properties have been explored for these channels, deepening the theoretical foundation and guiding the design of efficient error-correcting codes tailored to these reading models.

\subsection{Motivation}
Locally bounded functions are fundamental to the design and operation of function-correcting codes across various communication channels. Their importance motivates the development of function-correcting codes tailored to b-symbol read channels, which are critical in high-density data storage systems. A $b$-symbol read channel is one in which errors may affect multiple adjacent symbols at once, rather than just individual symbols. This models realistic scenarios where reading or writing errors occur over clusters of data points (for example, due to interference or defects in storage media). Singh et al. \cite{Singh2025} addressed this challenge by designing function-correcting codes specifically for such channels, meaning that their codes are structured not only to detect and correct random errors but also to efficiently correct errors that affect blocks of data simultaneously.

In this paper, we introduce the concept of locally $(\lambda,\rho,b)$-functions, where the parameter $\lambda$ defines the locality or the number of symbols accessed for correction, $\rho$ represents the allowable error tolerance within the local region, and $b$  characterizes the size of symbol blocks affected by correlated errors. This extended framework enables efficient and resilient function correction in environments where errors are clustered or non-uniformly distributed, such as in high-density data storage and communication systems. By integrating block-based error patterns into the locality model, locally $(\lambda,\rho,b)$-functions provide a flexible and practical approach to designing function-correcting codes that are both redundancy-efficient and capable of addressing practical error environments.

\subsection{Related Works}
Lenz et al. introduced the concept of function-correcting codes (FCCs) in \cite{Lenz2023}, where it was shown that FCCs are equivalent to irregular-distance codes. This equivalence highlights that the ability of an FCC to protect a specific function of the message corresponds to designing codes with varying error-correction capabilities across different parts of the message. Xia et al. \cite{Xia2023} extended the concept of function-correcting codes from binary symmetric channels to symbol-pair read channels, introducing a new coding paradigm adapted to settings where data are accessed in overlapping symbol pairs rather than individual symbols. In \cite{Singh2025}, Singh et al. advanced the theory by developing function-correcting codes for $b$-symbol  read channels over finite fields and provided the notion of irregular $b$-distance codes. 
Optimal redundancy plays a crucial role in the design and analysis of function-correcting codes. In \cite{Lenz2023,Premlal2024,Singh2025,Xia2023}, the authors provided either exact values for optimal redundancy or established upper and lower bounds for it, offering key insights into the efficiency limits of FCCs. The optimal redundancy of function-correcting codes for the Hamming weight function and the Hamming weight distribution function is discussed in \cite{Ge2025}, where the authors improved the known bounds established in \cite{Lenz2023}. In \cite{Ly2025}, Hoang Ly and Emina Soljanin demonstrated that the bound on optimal redundancy is tight for sufficiently large finite fields, establishing it as an upper bound in such cases. Building on this result, they proposed an encoding scheme that achieves the optimal redundancy. Motivated by these two limiting cases over the binary field, they further conjectured that the bound remains a valid upper bound across all finite fields. In Table \ref{tabredundancy}, we briefly list known bounds on the optimal redundancy of function-correcting codes for locally bounded functions.

    \vspace{2mm}
    \begin{table}[h!]
    \centering
    \small
    \renewcommand{\arraystretch}{1.5}
    \resizebox{\textwidth}{!}{
    \begin{tabular}{|c|c|c|c|c|}
    \Xhline{2\arrayrulewidth}
    \makecell{\textbf{Metric}} & \textbf{Locally $(\lambda, \rho)$-function} & \textbf{Optimal Redundancy} & \textbf{References} & \textbf{Remark}\\
    \Xhline{2\arrayrulewidth}
    \multirow{4}{*}{Hamming metric over $\mathbb{F}_2$} & Locally $(2, 2t)$-function & $r_f^H(k,t) = 2t$ & \cite[Lemma 5]{Lenz2023} & -\\
    \cline{2-5}
                         & Locally $(3, 2t)$-function & $r_f^H(k,t) = 3t$ & Theorem \ref{optimality_condition} & Our contribution\\
    \cline{2-5}
                         & Locally $(4, 2t)$-function & $r_f^H(k,t) = 3t$ & \cite[Theorem 6]{Rajput2025} & -\\ 
    \cline{2-5}
                         & Locally $(\lambda, 2t)$-function  & $r_f^H(k,t) \leq \lambda t$ & Theorem \ref{le_bound_lambda} & Our contribution\\                        
    \Xhline{3\arrayrulewidth}
    Homogeneous metric over $\mathbb{Z}_{2^s}$                & Locally $(2, 2t)_h$-function & $t\leq r_f^h(k,t) \leq 2t$ & \cite[Lemma 5]{liu2025function} & -\\
    \Xhline{3\arrayrulewidth}
   \multirow{2}{*}{Lee metric over $\mathbb{Z}_m$, $m\geq 2$}  & Locally $(2, 2t)_L$-function & $r_f^L(k,t) = \left\lceil \frac{2t}{\left\lfloor \frac{m}{2} \right\rfloor} \right\rceil $ & \cite[Proposition 5.6]{Verma2025FCCL} & - \\ 
    \cline{2-5}
                        & Locally $(\lambda, 2t)_L$-function & $r^L_f(k,t) \leq \left\lceil \frac{t}{\left\lfloor \frac{m}{2\lambda} \right\rfloor} \right\rceil $, $\lambda\leq \frac{m}{2}$ &\cite[Lemma 16]{rajan2025explicit} & -\\
    \cline{2-5}
                        & Locally $(3, 2t)_L$-function & $r_f^L(k,t) = t,$ $m=6$ &\cite[Lemma 17]{rajan2025explicit} & -\\
    \Xhline{3\arrayrulewidth}
    \multirow{4}{*}{$b$-Symbol metric}  & Locally $(2, 2t)$-function, $b=2$ & $2t - 2\leq r_f^b(k,t) \leq 2t - 1$ & \cite[Lemma 14]{Xia2023} & -\\
    \cline{2-5}
                & Locally $(2, 2t)$-function, $b\geq 1$ & $ 2(t- b + 1) \leq r_f^b(k,t) \leq 2t - b + 1 $ & \cite[Lemma 5.1]{Singh2025} & -\\         
    \cline{2-5}   
                & Locally $(4, 2t)$-function, $b\geq 1$ & $ r_f^b(k, t) \leq 3t-b+1$ & Lemma \ref{lambda4} & Our contribution\\  

    \cline{2-5} 
                & Locally $(2^b, 2t)$-function, $b\geq 1$ & $2(t - b + 1) \leq r_f^b(k, t) \leq 2t$ & Corollary \ref{opt_red_bdd} & Our contribution\\  
    \Xhline{3\arrayrulewidth}  
    
    \end{tabular}}
    \caption{Bounds on optimal redundancy of function-correcting codes under various metrics for locally bounded functions.}\label{tabredundancy}
    \end{table}
\subsection{Our Contributions}
The concept of locally ($\lambda,\rho$)-function-correcting codes is explored in \cite{Rajput2025}, where the authors propose an upper bound on the redundancy of such codes. This bound is derived on the basis of the minimum possible length of an error-correcting code for a given number of codewords and a specified minimum distance. Building on this idea, we generalize the concept of locally ($\lambda,\rho$)-function-correcting codes accommodating $b$-symbol read channels, and refer to the resulting framework succinctly as locally $(\lambda, \rho, b)$-functions. 

In this work, we introduce the concept of locally bounded functions in $b$-symbol read channel, called locally $(\lambda,\rho,b)$-function (see Definition \ref{deflocally}). The definition is a natural generalization of the corresponding one in the Hamming metric. We investigate various recurrence relations associated with locally $(\lambda, \rho, b)$-functions, with particular focus on how the parameters $\lambda$, $\rho$,  and $b$ interact and also determine the smallest value of $\lambda$ for which a given function $f$ qualifies as a locally ($\lambda,\rho, b$)-function. In particular, we systematically study the locality of the weight functions and weight distribution functions in the Hamming metric (see Theorems \ref{lambda_w}, \ref{lamhwd}) and $b$-symbol metric by providing the upper and lower bounds on the smallest value of $\lambda$ (see Corollaries \ref{luboundw}, \ref{luboundwd}). In addition, we present several examples that achieve the upper bound as well as the intermediate values for various choices of parameters $k,\lambda,$ and $\rho$ in Table \ref{tab1}.

In the later section, we derive several bounds on the optimal redundancy of locally $(\lambda, \rho, b)$-function-correcting codes by relating it to the minimum achievable length of $b$-symbol error-correcting codes, as well as classical error-correcting codes under the Hamming metric, for a fixed number of codewords and minimum distance in Theorems \ref{irreb} and \ref{genlambda}. In particular, we also present explicit upper bounds by constructing $(f,t)$-FCBSCs for the class of locally $(4, 2t, b)$-functions and locally $(2^b, 2t, b)$-functions in Lemma \ref{lambda4} and Theorem \ref{2^b_2t}, respectively. Furthermore, for the case $b = 1$, we show that a locally $(3, 2t, b)$-function achieves the optimal redundancy of $3t$ under certain conditions in Theorem \ref{optimality_condition}.  We also provide an explicit construction of an $(f,t)$-FCBSC with redundancy $2t-b+1$, where $f$  is the $b$-symbol weight distribution function in Theorem \ref{explicitcons}.

\subsection{Organization}
The paper is organized as follows. In Section \ref{pre}, we recall the fundamentals and key results related to $b$-symbol codes and function-correcting codes. In Section \ref{sectionlocally}, we introduce the concept of locally bounded functions for $b$-symbol read channels and investigate the values of parameters for which a given function behaves like a locally bounded function. In Section \ref{sectionredundancy}, we establish several bounds on the optimal redundancy of an $(f,t)$-FCBSCs for the class of locally bounded functions. Finally, in Section \ref{sectionconclusion}, we conclude the paper.

\subsubsection*{Notations}
    Throughout this paper, we use the following notations. Let $\mathbb{F}_2$ denote the binary field and $\mathbb{F}_2^n$ be the set of all binary vectors of length $n$. We also employ the standard conventions where $\mathbb{N}_{0}$ represents the set of all non-negative integers, $[M]$ denotes the set of positive integers $\{1, 2, \ldots, M\}$, and $a^+$ denotes $\max\{a,0\}$. For a positive real number $a$, $\{a\}$ denotes the fractional part of $a$, that is, $\{a\}=a-[a]$, where $[a]$ is integer part of $a$. The code parameters are the message length $k$, the redundancy $r$, and the total code length $n = k + r$; the parameter $ b \ge 1$ denotes the size of the read window in the $b$-symbol channel. For a vector $x = (x_0, x_1, \dots, x_{n-1}) \in \mathbb{F}_2^n$, its $b$-symbol read vector is denoted by $\pi_b(\mathbf{x})$. The Hamming and $b$-symbol distances between two vectors $x$ and $y$ are written as $d_H(x,y)$ and $d_b(x,y)$, respectively, with their corresponding weights denoted by $w_H(x)$ and $w_b(x)$. A ball of radius $\rho$ centered at $u$ is represented as $B_H(u,\rho)$ for the Hamming metric and $B_b(u,\rho)$ for the $b$-symbol metric. For a function $f : \mathbb{F}_2^k \to \text{Im}(f)$, the function ball of radius $\rho$ in the $b$-symbol metric is denoted by $B^b_f(u,\rho)$, and its locality is characterized by the parameters $(\lambda, \rho, b)$ for some positive integer $\lambda$, we denote the smallest $\lambda$ for which $f$ is a locally $(\lambda, \rho, b)$-function by $\lambda_s$. An $(f,t)$-function-correcting $b$-symbol code is abbreviated as an $(f,t)$-FCBSC, and its optimal redundancy is $r^b_f(k,t)$. Finally, $N_b(\lambda,\rho)$ and $N_H(\lambda,\rho)$ denote the minimum length of a code with $\lambda$ codewords and minimum distance $\rho$ in the $b$-symbol and Hamming metrics, respectively, while $\Delta^b_T$ and $\Delta^H_T$ represent their respective weight distribution functions with a threshold $T$.

\section{Preliminaries}\label{pre}
This section presents fundamental concepts related to classical $b$-symbol codes and function-correcting $b$-symbol codes. These definitions and properties form the foundational tools necessary for understanding the main results developed in the subsequent sections of this paper. For  $b=1$, function-correcting $b$-symbol codes reduce to the classical function-correcting codes (FCCs). We use the term FCBSC in place of FCC to maintain generality. To align with existing notation for function-correcting codes in the Hamming metric, we denote redundancy by \( r_f^H(k, t) \).

\subsection{Codes in the $b$-symbol metric}
In this subsection, we briefly recall the fundamentals of $b$-symbol codes. We begin by defining the $b$-symbol distance and then present several relations between the Hamming distance and the $b$-symbol distance. For this, we refer \cite{Chen2025,Ding2018,Luo2024,b_symbol,Singh2025}. 
\begin{definition}[\cite{b_symbol} $b$-symbol read vector]
Let $b\geq 1$ be an integer and $x = (x_0, x_1, \ldots, x_{n-1})$ be a vector in $\mathbb{F}_2^n$. The \emph{$b$-symbol read vector} of $x$, denoted by $\pi_b(x)$, is defined as
\begin{equation*}
    \pi_b(x) = \left( (x_0, x_1, \ldots, x_{b-1}), (x_1, x_2, \ldots, x_b), \ldots, (x_{n-1}, x_0, \ldots, x_{b-2}) \right),
\end{equation*}
where each entry is a consecutive sequence of $b$ symbols taken cyclically from $x_0$.
\end{definition}

\begin{definition}[\cite{b_symbol} $b$-Symbol distance]
Given two vectors $x, y \in \mathbb{F}_2^n$, the \emph{$b$-symbol distance} between them is defined by
\begin{equation*}
    d_b(x, y) = d_H\left(\pi_b(x), \pi_b(y)\right),
\end{equation*}
where $d_H$ denotes the Hamming distance and $(x_{i},x_{i+1},\dots,x_{i+b-1})\neq (y_{i},y_{i+1},\dots,y_{i+b-1})$ if $x_{i+j}\neq y_{i+j}$ for some $0\leq j\leq b-1$.
\end{definition}
\begin{definition}[\cite{b_symbol} $b$-Symbol read code]
Let $\mathcal{C}$ be a code over $\mathbb{F}_2$. The corresponding \emph{$b$-symbol read code}, denoted by $\pi_b(\mathcal{C})$, is the image of $\mathcal{C}$ under the $b$-symbol read mapping:
\begin{equation*}
    \pi_b(\mathcal{C}) = \left\{ \pi_b(c) : c \in \mathcal{C} \right\}.
\end{equation*}
The minimum $b$-symbol distance of $\mathcal{C}$ is defined as
\begin{equation*}
d_b(\mathcal{C}) = d_H\left(\pi_b(\mathcal{C})\right),  
\end{equation*}
that is, the minimum Hamming distance between any two distinct elements of $\pi_b(\mathcal{C})$. 
\end{definition}
Let $u\in\mathbb{F}_2^n$. Define a ball $B_b(u,\rho)$  with center  $u$ and radius $\rho$ as follows
$$B_b(u,\rho)=\{y\in\mathbb{F}_2^n:\ d_b(u,y)\leq \rho\}.$$
For $b=1$, the $b$-symbol distance coincides with the Hamming distance. We denote a ball in the Hamming metric by $B_H(u,\rho)$ with center  $u$ and radius $\rho$.
We now recall some results that will be used later in the paper.
\begin{lemma}\cite[Lemma 2.1]{Singh2025}\label{handb}
 Let $x,y\in \mathbb{F}_2^n$. Then
 \begin{enumerate}
     \item If $d_H(x,y)>n-(b-1)$, then $d_b(x,y)=n$.
     \item If $0<d_H(x,y)\leq n-(b-1)$, then $d_H(x,y)+(b-1)\leq d_b(x,y)\leq b\cdot d_H(x,y)$.
     \item If $d_H(x,y)=0$, then $d_b(x,y)=0$.
 \end{enumerate}
\end{lemma}

\begin{lemma}\cite[Proposition 2.2]{Singh2025}\label{bandb+1}
For any two vectors $x = (x_0, \dots, x_{n-1})$ and $y = (y_0, \dots, y_{n-1})$ in $\mathbb{F}_2^n$ 
and $0 < d_b(x,y) < n$, we have
\[
d_{b+1}(x,y) \geq d_b(x,y) + 1.
\]
\end{lemma}

\begin{lemma}\cite[Lemma 1]{Luo2024}\label{tuplequal}
Let $x=(x_0,x_1,\dots,x_{n-1})$, $y=(y_0,y_1,\dots,y_{n-1})\in \mathbb{F}_2^n$. If $n\geq b$ and $x_j=y_j$ for each $j\in [0,b-2]$, then $$w_b((x,y))=w_b(x)+w_b(y).$$    
\end{lemma}
Using similar arguments as in the above lemma, for $x=(x_0,x_1,\dots,x_{n-1})\in \mathbb{F}_2^n$ and $n\geq b$,  we have \begin{align}\label{mcopies}
     w_b(\underbrace{x,x,\dots,x}_{\text{ $m$ copies}})=m\cdot w_b(x).
 \end{align}

\begin{lemma}\cite[Lemma 3.1]{Singh2025}\label{lemma3.1frm12}
 Let $x=(x^{(1)},x^{(2)})$ and $y=(y^{(1)},y^{(2)})$ in $\mathbb{F}_2^{k+r}$, where 
 \begin{align*}
     x^{(1)}=(x_0,x_1,\dots,x_{k-1}), &\ \ \  x^{(2)}=(x_k,x_{k+1},\dots,x_{k+r-1}),\\
      y^{(1)}=(y_0,y_1,\dots,y_{k-1}), &\ \ \  y^{(2)}=(y_k,y_{k+1},\dots,y_{k+r-1}).
 \end{align*}
 Then
 \begin{align*}
     d_b(x^{(1)},y^{(1)})+d_b(x^{(2)},y^{(2)})-(b-1)\leq d_b(x,y)\leq d_b(x^{(1)},y^{(1)})+d_b(x^{(2)},y^{(2)})+(b-1).
 \end{align*}
\end{lemma}

\subsection{Function-Correcting Codes for $b$-Symbol Read Channels}
Function-correcting codes were introduced by Lenz et al. \cite{Lenz2023} in the context of the Hamming metric. Later, in \cite{Singh2025}, the authors discussed a broader framework for function correcting codes to $b$-symbol read channels. Here, we briefly recall the basic definitions and key results which are used throughout the paper. For detailed proofs and discussions, we refer \cite{Lenz2023,Singh2025}. 
\begin{definition}\cite{Singh2025}[Function-correcting $b$-symbol code]
An encoding function
\[
\text{Enc}: \mathbb{F}_2^k \rightarrow \mathbb{F}_2^{k+r}, \quad \text{Enc}(x) = (x, p(x))\  \forall x\in \mathbb{F}_2^k
\]
is said to define an $(f,t)$-function-correcting $b$-symbol code (in short FCBSC) for a positive integer $t$ and   a function

$$f: \mathbb{F}_2^k \rightarrow \text{Im}(f),$$

if for all $x_1, x_2 \in \mathbb{F}_2^k$ with $f(x_1) \neq f(x_2)$, the following holds
\[
d_b(\text{Enc}(x_1), \text{Enc}(x_2)) \geq 2t + 1.
\]
\end{definition}
The set $C=\{Enc(x):\ x\in \mathbb{F}_2^k\}\subseteq \mathbb{F}_2^{k+r}$ is said to be a code over $\mathbb{F}_2$ with a redundancy $r$ (equivalently, of length $k+r$). The elements of $C$ are called codewords.
By definition, an $(f,t)$-FCBSC is capable of correcting up to $t$  errors occurring in codewords associated with distinct function values. While an error-correcting $b$-symbol code $C$ is capable of correcting up to $t$ errors if and only if the $b$-symbol distance between any two distinct codewords is at least $2t+1$, that is, $d_b(C) \geq 2t + 1$ \cite{b_symbol}. Therefore, systematic error-correcting $b$-symbol codes with $t$ error correction capability, in which the first $k$ bits correspond to message bits, inherently serve as an $(f,t)$-FCBSC for all functions $f:\mathbb{F}_2^k\to \text{Im} (f)$.

\begin{definition}\cite{Singh2025}[Optimal redundancy]
Let $t$ be a positive integer and $f:\mathbb{F}_2^k\to \text{Im}(f)$ be a function.
The optimal redundancy of an $(f,t)$-FCBSC is the smallest positive integer $r$ for which there exists an $(f,t)$-function-correcting $b$-symbol code with redundancy $r$. It is denoted by $r_f^b(k, t).$
\end{definition}
\begin{definition}\cite{Singh2025}[Irregular $b$-symbol distance matrices]
Let $f:\mathbb{F}_2^k \to \text{Im}(f)$ be a function. Consider $M$ vectors $x_1, \ldots, x_{M} \in \mathbb{F}_2^k$. Then, the matrices $B_f^{(1)}(t, x_1, \ldots, x_{M})$ and $B_f^{(2)}(t, x_1, \ldots, x_{M})$ are $M \times M$ irregular $b$-symbol distance matrices defined as follows
\[
[B_f^{(1)}(t, x_1, \ldots, x_{M})]_{ij} =
\begin{cases}
[2t - b + 2 - d_b(x_i, x_j)]^+ & \text{if } f(x_i) \neq f(x_j), \\
0 & \text{otherwise},
\end{cases}
\]
and 
\[
[B_f^{(2)}(t, x_1, \ldots, x_{M})]_{ij} =
\begin{cases}
[2t + b - d_b(x_i, x_j)]^+ & \text{if } f(x_i) \neq f(x_j), \\
0 & \text{otherwise}.
\end{cases}
\]
\end{definition}
In the Hamming metric ($b=1)$, the matrices $B_f^{(1)}$ and $B_f^{(2)}$ coincide and reduce to the distance requirement matrix introduced in \cite{Lenz2023} and defined as follows.
    \begin{definition}\cite{Lenz2023}[Distance Requirement Matrix]
        Let $x_1, \ldots, x_{M} \in \mathbb{Z}_2^k$.  
        The \emph{distance requirement matrix} of a function $f$, denoted by 
        $D_f(t, x_1, \ldots, x_{M})$, is the $M \times M$ matrix whose $(i,j)$-th entry is defined as
        \[
        [D_f(t, x_1, \ldots, x_{M})]_{ij} =
        \begin{cases}
        \big[\, 2t + 1 - d_H(x_i, x_j) \,\big]^+, & \text{if } f(u_i) \ne f(u_j), \\[4pt]
        0, & \text{otherwise}.
        \end{cases}
        \]      
    \end{definition}

\begin{definition}\cite{Singh2025}[$B_b$-Code]
A set of codewords $\mathcal{P} = \{ \boldsymbol{p}_1, \ldots, \boldsymbol{p}_M \}$ is said to be a {$B$-irregular $b$-symbol distance code} (in short, {$B_b$-code}) for some matrix $B \in \mathbb{N}_0^{M \times M}$ if there exists an ordering of the codewords in $\mathcal{P}$ such that
\[
d_b(\boldsymbol{p}_i, \boldsymbol{p}_j) \geq [B]_{ij} \quad \text{for all } i, j \in \{1, 2, \ldots, M\}.
\]
\end{definition}

Let $N_b(B)$ denote the smallest positive integer $r$ for which there exists a $B_b$-code of length $r$. The following lemma gives a generalization of the Plotkin bound for irregular distance codes in the Hamming metric.
\begin{lemma}\cite[Lemma 1]{Lenz2023}\label{lemma1frm7}
 Let $b=1$ (the Hamming metric). For any distance matrix  $B_f^{(1)}$  of order $M\times M$ over positive integers, we have 
 \begin{align*}
     N_H(B_f^{(1)})\geq \begin{cases}
         \frac{4}{M^2}\sum_{i,j;i<j}[B_f^{(1)}]_{ij} & \text{ if } M \text{is even},\\
         \frac{4}{M^2-1}\sum_{i,j;i<j}[B_f^{(1)}]_{ij} & \text{ if } M \text{is odd}.
     \end{cases}
 \end{align*}
\end{lemma}

\begin{lemma}\cite[Lemma 4.6]{Singh2025}\label{nbnh}
    Let $\lambda$ and $t$ be positive integers with $\lambda > 2^{b-1}$
and $2t \geq b$. Then,
\begin{equation*}
    N_b(\lambda,2t) \leq N_H(\lambda,2t-b+1).
\end{equation*}
\end{lemma}

\begin{theorem}\cite[Theorem 7]{Rajput2025}\label{thm7frm11}
 Let $t$ be a positive integer. For any locally $(\lambda,2t,b=1)$-function $f$, the optimal redundancy of an $(f,t)$-FCC is bounded by 
 $$r_f^H(k,t)\leq N_H(\lambda,2t).$$
\end{theorem}

\begin{corollary}\label{cor3.3frm12}\cite[Corollary 3.3]{Singh2025}
Let $t$ be a positive integer and $f:\mathbb{F}_2^k\to \text{Im}(f)$ be a function with $|\text{Im}(f)|\geq 2$ and $t>b-1$. Then, $r_f^b(k,t)\geq 2(t-b+1)$.    
\end{corollary}



\section{Locally $\mathbf(\lambda,\rho,b)$-function}\label{sectionlocally}

In this section, we introduce the concept of locally $(\lambda,\rho)$-functions over $b$-symbol read channels, which we refer to succinctly as locally $(\lambda,\rho,b)$-functions. Every function $f$ exhibits $(\lambda,\rho,b)$-locality for appropriate choices of $\lambda$, $\rho$, and $b$. We investigate various recurrence relations associated with $(\lambda,\rho,b)$-local functions, focusing on how these parameters interact. We improve several results established in \cite{Rajput2025} and further generalize them to the setting of $b$-symbol read channels. First, we formally define the notion of a locally $(\lambda,\rho,b)$-function.

A function ball is the collection of function values that lie within a specified radius, measured with respect to the $b$-symbol metric, around a given point. Formally defined as follows.
\begin{definition}\cite{Singh2025}[Function Ball]
Let \( f : \mathbb{F}_2^k \to \operatorname{Im}(f) \) be a function. For a vector \( u \in \mathbb{F}_2^k \) and a non-negative integer \( \rho \), the \emph{function ball} of \( f \) with radius \( \rho \) centered at \( u \) is defined as
\[
B_f^b(u, \rho) = \{ f(u') \mid u' \in \mathbb{F}_2^k \text{ and } d_b(u, u') \leq \rho \}.
\]
For $b=1$, a function ball in the Hamming metric is denoted by $B_f(u,\rho)$ and defined as
$$B_f(u, \rho) = \{ f(u') \mid u' \in \mathbb{F}_2^k \text{ and } d_H(u, u') \leq \rho \}.$$
\end{definition}

\begin{definition}\label{deflocally}
 A function $f:\mathbb{F}_2^k\to \text{Im}(f)$ is said to be a  locally $(\lambda, \rho, b)$-function if $$|B_f^b(u,\rho)|\leq \lambda$$ for all $u\in \mathbb{F}_2^k$.
\end{definition}
Determining the smallest value of $\lambda$ for which a given function $f$ is a locally $(\lambda,\rho,b)$-function, for fixed values of $\rho$ and $b$, presents a nontrivial challenge. It is easy to observe that any function is $(\lambda,\rho,b)$-local when $\lambda = |\text{Im}(f)|$, irrespective of the choices of $\rho$ and $b$. Similarly, when $\rho = k$, the function $f$ is  $(\lambda,\rho,b)$-local only if $\lambda \geq |\text{Im}(f)|$. Also,  for $2 \leq \beta \leq \lambda$, every locally $(\beta, \rho,b)$-function is a locally $(\lambda, \rho, b)$-function. Since these scenarios are relatively straightforward and offer limited theoretical insights, our attention is primarily directed toward the study of locally $(\lambda,\rho,b)$-functions under the more restrictive conditions $\lambda \leq |\text{Im}(f)|$ and $\rho \leq k$.  We denote $\lambda_s$ to be the smallest value of $\lambda$ for which a function $f$ is a locally $(\lambda,\rho,b)$-function. Note that 
$|B_f^b(u,\rho)|\leq |B_b(u,\rho)|$ for all $u \in \mathbb{F}_2^k$. In fact, for all $u\in\mathbb{F}_2^k$, we have  $|B_b(u,\rho)|= |B_b(\mathbf{0},\rho)|$. This observation gives the following result.
\begin{corollary}\label{balllocal}
   Let $f:\mathbb{F}_2^k\to \text{Im}(f)$ be a locally $(\lambda_s,\rho,b)$-function. Then $\lambda_s\leq |B_b(\mathbf{0},\rho)|$.
\end{corollary}

    \begin{proposition}\label{recb}
         If $2\leq \rho+1\leq k-1$ and $f:\mathbb{F}_2^k\to \text{Im}(f)$ is a locally $(\lambda, \rho, b)$- function, then $f$ is a locally $(\lambda, \rho+1,b+1)$-function.  
    \end{proposition}
    \begin{proof}
        Let $f$ be a locally $(\lambda, \rho, b)$-function. By definition, we have 
        $|B^b_f(u, \rho)| \le \lambda$ for all $ u \in \mathbb{F}_2^k.$ If $f(v) \in B^{b+1}_f(u, \rho+1)$, then $d_{b+1}(u,v) \le \rho+1$.  By Lemma \ref{bandb+1}, $d_b(u,v) \le d_{b+1}(u,v) - 1 \le \rho$,  
        therefore $f(v) \in B^b_f(u, \rho)$. Thus $B^{b+1}_f(u, \rho+1) \subseteq B^b_f(u, \rho)$, and hence 
        \[
        |B^{b+1}_f(u, \rho+1)| \le |B^b_f(u, \rho)| \le \lambda.
        \]
        Therefore, $f$ is a locally $(\lambda, \rho+1, b+1)$-function.
    \end{proof}

The next proposition establishes a connection between the containment properties of balls under the Hamming distance and those under the $b$-symbol distance. 
\begin{proposition}\cite[Proposition 2]{Chen2025}.
    $B_b(u,\rho)\subset B_H(u,\rho-b+1)$ and $B_H(u,\rho)\subset B_b(u,b\rho)$. Also, $B_b(u,\rho)=\{u\}$ for $\rho\leq b-1$ and $B_b(u,b)=B_H(u,1)$.
\end{proposition}
The following results arise immediately as consequences of the above proposition.


\begin{corollary}
  $B^b_f(u,\rho)\subset B^H_f(u,\rho-b+1)$ and $B^H_f(u,\rho)\subset B^b_f(u,b\rho)$. Moreover, $B^b_f(u,\rho)=\{f(u)\}$ for $\rho\leq b-1$ and $B^b_f(u,b)=B^H_f(u,1)$.   
\end{corollary}

\begin{corollary}\label{recball}
    Let $f:\mathbb{F}_2^k\to \text{Im}(f)$ be a function.  
    \begin{enumerate}
        \item If $f$ is  a locally $(\lambda, \rho, b=1)$-function, then $f$ is a locally $(\lambda,\rho+b-1,b)$-function.
        \item If $f$ is  a locally $(\lambda, b\rho, b)$-function, then $f$ is  a locally $(\lambda,\rho,b=1)$-function.
        \item For, $\rho\leq b-1$, $f$ is locally $(1,\rho,b)$-function. 
        \item $f$ is a locally $(\lambda,b,b)$-function if and only if $f$ is a locally $(\lambda,\rho=1, b=1)$-function.
    \end{enumerate}
\end{corollary}
\subsection{Locality of the weight distribution function}
In this subsection, we investigate weight distribution functions associated with both the Hamming distance and the $b$-symbol distance. Let $f:\mathbb{F}_2^k\to \mathbb{Z}^{\geq0}$ be the $b$-symbol weight distribution function with threshold $T$, defined as 
$f(u)=\Delta_T^b(u)=\left \lfloor\frac{w_b(u)}{T}\right \rfloor$.
 In the special case when $T=1$, the function reduces to $\Delta^b(u) = w_b(u)$, known as the $b$-symbol weight function. Furthermore, when $b=1$, the corresponding function is referred to as the Hamming weight distribution function, denoted by $\Delta_T^H$. In particular, for $b=1$ and $T=1$, it simplifies to the Hamming weight function, denoted by $w_H$.

\begin{proposition}
Let $\Delta_T^b:\mathbb{F}_2^k\to \text{Im}(f)$ be the $b$-symbol weight distribution function. Then $\Delta_T^b$ is a locally $(\left \lfloor\frac{k}{T}\right\rfloor+1,\rho,b)$-function for any $0 \leq \rho\leq k$. That is $\lambda_s\leq \left \lfloor\frac{k}{T}\right \rfloor+1$. 
\end{proposition}
\begin{proof}
The proof follows from the fact $0\leq \Delta_T^b(v)\leq \left \lfloor \frac{k}{T}\right \rfloor$ for all $v\in \mathbb{F}_2^k$ and $0\leq \rho \leq k$.     
\end{proof}
The following theorem gives a value of $\lambda$ for which the Hamming weight function $w_H$ serves as a locally $(\lambda, \rho, b=1)$-function. The proof is based on finding the maximum and minimum achievable Hamming weight in an arbitrary function ball.
\begin{theorem}\label{lambda_w}
    For $k \geq 2\rho$, the Hamming weight function $w_H$ is a locally $(2\rho+1,\rho,b=1)$-function. 
\end{theorem}
\begin{proof}
    We claim that for $b=1$, $|B^b_{w_H}(u,\rho)|\leq 2\rho +1$ for all $u\in \mathbb{F}_2^k$. Let $u \in \mathbb{F}_2^k$ and $w_H(u) = w$. Also, let $u' \in \mathbb{F}_2^k$ such that $d_H(u,u') = d \leq \rho$. Then, $w_H(u') \in B_{w_H}^b(u,\rho) $ and can be expressed in terms of the Hamming weight of $u$ as follows
    \begin{equation}
        w_H(u') = w +a - b,   
    \end{equation}
    where $a$ is the number of changes of symbols in $u$ from $0$ to $1$ and $b$  is the number of changes of symbols in $u$ from $1$ to $0$ to obtain $u'$ from $u$.
    As $a + b = d$, the Hamming weight of any arbitrary  $u'$ such that $ w_H(u') \in B_{w_H}^b(u,\rho)$  in terms of $d$ can be given as 
    \begin{equation}\label{wt2}
        w_H(u') = \{w + d - 2b : 0 \leq d \leq \rho \text{ and } 0 \leq b \leq \min(d,w)\}.
    \end{equation}
    From \ref{wt2}, we can see that the Hamming weight inside the function ball ranges over the interval 
    \begin{equation}\label{wt}
        [\max(0,w-\rho), \min(k,w+\rho)].
    \end{equation} 
   The length of this interval depends on $w$.\\
    \textbf{Case 1:} If $\rho \leq w \leq k-\rho$, then  Interval \ref{wt} becomes $[{w - \rho}, {w + \rho}]$ and the length of this interval is $2\rho + 1$.\\ 
    \textbf{Case 2:} If $\rho > w $,  then Interval  \ref{wt} becomes $[0, {w + \rho}]$ and the length of this interval is $w+\rho+1 <2\rho + 1$.\\ 
    \textbf{Case 3:} If  $w > k-\rho,$ then Interval \ref{wt} becomes $[{w - \rho}, k]$ and the length of this interval is $k - w + \rho + 1 < 2\rho +1 $.\\ 
    Hence, from all the above cases, we get $|B^b_{w_H}(u,\rho)|\leq 2\rho +1$ for all $u\in \mathbb{F}_2^k$. This completes the proof. Moreover, the bound is tight as there always exists a $u\in \mathbb{F}_2^k$ such that $\rho\leq w_H(u)\leq k-\rho$.
\end{proof} 

The following example visualizes the function ball for all three cases discussed in Theorem \ref{lambda_w} for the Hamming weight function.
\begin{example}
    Consider the Hamming weight function 
\[
w_H : \mathbb{F}_2^6 \rightarrow \{0,1,2,3,4,5,6\}
\]
The table shows all possible Hamming weights attained in the function ball $B_{w_H}(u,\rho)$ for the radius $\rho=2$ for three different binary vectors  belonging to three cases as discussed in Theorem \ref{lambda_w}\\

\scriptsize $\begin{array}{|c|c|c|c|c|c|}
\Xhline{1.2pt}
\# & d & a & \text{vectors } u' & \text{wt}(u') & |B_{wt}(u,2)| \\
\Xhline{1.2pt}
 \makecell{Case 1\\  u = 000111} & 0 & 0 & u' = 000111 & 3 & 5\\
\cline{2-5}
   & 1 & 0 & u' = \{000011, 000101, 000110\} & 2 &\\
\cline{3-5}
   &   & 1 & u' = \{001111, 010111, 100111\} & 4 &\\
\cline{2-5}
   & 2 & 0 & u' = \{000001, 000100, 000010\} & 1 &\\
\cline{3-5}
   &   & 1 & \text{no change} & - &\\
\cline{3-5}
   &   & 2 & u' = \{011111, 101111, 110111\} & 5 &\\
\Xhline{1.2pt}
\makecell{Case 2 \\ u = 100000} & 0 & 0 & u' = 100000 & 1 & 4 \\
\cline{2-5}
   & 1 & 0 & u' = \{000000\} & 0 &\\
\cline{3-5}
   &   & 1 & u' = \{110000, 101000, 100100, 1000010,100001\} & 2 & \\
\cline{2-5}
   & 2 & 0 & - & - &\\
\cline{3-5}
   &   & 1 & \text{no change} & - &\\
\cline{3-5}
   &   & 2 & u' = \makecell{\{111000,110100,110010,110001,101100,101010, \\ 101001,100110,100110,100101,100011\}} & 3 &\\
\Xhline{1.2pt}
\makecell{Case 3 \\ u = 111111} & 0 & 0 & u' = 111111 & 6 & 3\\
\cline{2-5}
   & 1 & 0 & u' = \{011111, 101111, 110111, 111011, 111101, 111110\} & 5 &\\
\cline{3-5}
   &   & 1 & - & - &\\
\cline{2-5}
   & 2 & 0 & \makecell{u' = \{001111, 010111, 011011, 011101, 011101,\\ 011110, 100111, 101011, 101101, 101110, 110011, 110101,\\ 110110, 111001, 111010, 111100\}} & 4 &\\
\cline{3-5}
   &   & 1 & \text{no change} & - &\\
\cline{3-5}
   &   & 2 & - & - &\\
\Xhline{1.2pt}
\end{array}\label{hamming_weight}
$
\end{example}

In \cite[Corollary 2]{Rajput2025}, the authors showed that the Hamming weight function $w_H$ is a locally  $(4t+2,2t,b=1)$-function. Theorem \ref{lambda_w} generalizes their result for an arbitrary positive integer $\rho$. Moreover, Theorem \ref{lambda_w} improves their result as stated in the following corollary.

\begin{corollary}
    The Hamming weight function $w_H$ is a locally $(4t+1,2t,b=1)$-function.
\end{corollary}

 Next, we prove that the Hamming weight distribution function $\Delta_T^H$ is a locally $(\lambda,\rho,b=1)$-function for a value of $\lambda$. The proof proceeds along the same lines as that of Theorem \ref{lambda_w}, with appropriate modifications. We provide the proof for completeness. For a positive real number $a$, $\{a\}$ denotes the fractional part of $a$, that is, $\{a\}=a-[a]$, where $[a]$ is integer part of $a$.
\begin{theorem}\label{lamhwd}
For $k \geq 2\rho$, the Hamming weight distribution function $\Delta_T^H$ is a locally $(\left\lfloor\frac{2\rho}{T}\right\rfloor+2,\rho,b=1)$-function. 
\end{theorem}
\begin{proof}
    From the proof of Theorem \ref{lambda_w}, we can see that the Hamming weight inside the function ball ranges over the interval $[\max(0,w-\rho), \min(k,w+\rho)]$. Hence, the Hamming weight distribution inside $B_{\Delta^{H}_T}^b(u,\rho)$ ranges over the interval
    \begin{equation}\label{wt_dis}
         \left[\left\lfloor \frac{\max(0,w-\rho)}{T} \right \rfloor, \left\lfloor\frac{\min(k,w+\rho)}{T} \right\rfloor\right].
    \end{equation}
   The length of this interval depends on $w$ as follows.\\
    \textbf{Case 1:}  If $\rho \leq w \leq k-\rho$, then  Interval \ref{wt_dis} becomes $\left[\left\lfloor \frac{w - \rho}{T} \right \rfloor, \left\lfloor \frac{w + \rho}{T} \right\rfloor\right]$ of length $\left\lfloor \frac{w + \rho}{T} \right\rfloor - \left\lfloor \frac{w - \rho}{T} \right \rfloor + 1 $. We consider two further subcases.
    \begin{itemize}
        \item If $\left\{ \frac{w +\rho}{T}\right\} < \left\{ \frac{w - \rho}{T}\right\}$,  then the length of Interval \ref{wt_dis} is $$\left\lfloor \frac{w + \rho}{T} \right\rfloor - \left\lfloor \frac{w - \rho}{T} \right \rfloor + 1  =  \left\lfloor \frac{w + \rho}{T} -  \frac{w - \rho}{T} \right \rfloor + 2 =  \left\lfloor \frac{2\rho}{T} \right \rfloor +2.$$
        \item If $\left\{ \frac{w + \rho}{T}\right\} \geq \left\{ \frac{w - \rho}{T}\right\}$,  then the length of Interval \ref{wt_dis} is $$\left\lfloor \frac{w + \rho}{T} \right\rfloor - \left\lfloor \frac{w - \rho}{T} \right \rfloor + 1  =  \left\lfloor \frac{w + \rho}{T} -  \frac{w - \rho}{T} \right \rfloor + 1 =  \left\lfloor \frac{2\rho}{T} \right \rfloor + 1.$$
    \end{itemize}
    \textbf{Case 2:} If $\rho > w $, then Interval \ref{wt_dis} becomes $\left[0, \left\lfloor\frac{w + \rho}{T}\right\rfloor\right]$ of length $$\left\lfloor \frac{w + \rho}{T} \right\rfloor + 1  < \left\lfloor \frac{2\rho}{T} \right \rfloor +1  < \left\lfloor \frac{2\rho}{T} \right \rfloor +2 .$$ 
    \textbf{Case 3:} If  $w > k-\rho,$ then Interval \ref{wt_dis} becomes $\left[\left\lfloor\frac{w - \rho}{T}\right\rfloor, \left\lfloor\frac{k}{T}\right\rfloor\right]$ of length $$\left\lfloor\frac{k}{T}\right\rfloor - \left\lfloor\frac{w - \rho}{T}\right\rfloor + 1 < \left\lfloor \frac{w+\rho}{T} - \frac{w-\rho}{T}\right\rfloor + 2 \leq \left\lfloor \frac{2\rho}{T} \right \rfloor +2.$$
    That is, the length of Interval \ref{wt_dis} is strictly less than  $ \left\lfloor \frac{2\rho}{T} \right \rfloor +2$.\\
    Thus, we have
    $$
    |B^b_{\Delta^{H}_T}(u,\rho)| \leq \left\lfloor \tfrac{2\rho}{T} \right\rfloor + 2
    \quad \text{for all } u \in \mathbb{F}_2^k.
    $$
     This completes the proof.
\end{proof}
Observe that the bound in Theorem \ref{lamhwd} is tight for $T>0$ such that $\bigl\{\tfrac{w+\rho}{T}\bigr\} < \bigl\{\tfrac{w-\rho}{T}\bigr\}$ for some $\rho \leq w\leq k-\rho$. In this case, the situation falls under the first subcase of Case 1 in the proof of Theorem \ref{lamhwd}. In the subsequent corollaries, we explicitly describe these cases for different values of $T$.


\begin{corollary}
    Let $k,\  \rho \in \mathbb{N}$ with $k > 2\rho$. Then the Hamming weight distribution function is a locally $\left(\lambda_s=\left\lfloor\frac{2\rho}{T}\right\rfloor+2,\rho,b=1\right)$-function for $T=2\rho+1$.
\end{corollary}
\begin{proof}
As $k>2\rho$, there exists $u\in \mathbb{F}_2^k$ with $w_H(u)=w=\rho+1$. Consequently, $\rho\leq w\leq k-\rho$ and 
\begin{align*}
    0=\left \{\frac{w+\rho}{T}\right \}< \left \{\frac{w-\rho}{T}\right \}=\frac{1}{T}.
\end{align*}
 Thus, $u$ falls under the first subcase of Case 1 in the proof of Theorem \ref{lamhwd} and we have  $$|B_{\Delta_T^H}^b(u, \rho)|=\left\lfloor\frac{2\rho}{T}\right\rfloor+2.$$
This completes the proof.
\end{proof}
    \begin{corollary}
        For $k, \rho, T \in \mathbb{N}$ with $k \geq 2\rho$, the Hamming weight distribution function is a locally $\left(\left\lfloor\frac{2\rho}{T}\right\rfloor+1,\rho,b=1\right)$-function if $T$ divides $\rho$ or $k < T.$
    \end{corollary}   
    \begin{proof}
        In Theorem \ref{lamhwd} case $1$, it is straightforward that the condition $\left\{ \frac{w + \rho}{T}\right\} \geq \left\{ \frac{w - \rho}{T}\right\}$ is satisfied when $T$ divides $\rho$. Hence, for all three cases, we get $|B^b_{\Delta^{H}_T}(u,\rho)|\leq \left\lfloor \frac{2\rho}{T} \right \rfloor + 1$, making the Hamming weight distribution function a locally $(\lfloor\frac{2\rho}{T}\rfloor+1,\rho,b=1)$-function. \\
        For $k < T$ and $\rho \leq w \leq k-\rho$, we have $0 \leq \frac{w-\rho}{T} \leq \frac{k-2\rho}{T} < 1$ and $\frac{2\rho}{T} \leq \frac{w+\rho}{T} \leq \frac{k}{T} < 1$. Thus, we have $$ \left\{\frac{w-\rho}{T}\right\} = \frac{w-\rho}{T} \text{ and } \left\{\frac{w+\rho}{T}\right\} = \frac{w+\rho}{T}.$$
        Hence $\left\{ \frac{w + \rho}{T}\right\} \geq \left\{ \frac{w - \rho}{T}\right\}$. Therefore, $\Delta^H_T$ is a locally $(\lfloor\frac{2\rho}{T}\rfloor+1,\rho,b=1)$-function.
    \end{proof}

    
Using the recurrence relation in Proposition \ref{recb}, we have the following corollaries.

\begin{corollary}
    The Hamming weight distribution function $\Delta_T^H$ is a locally $(\lfloor\frac{2\rho}{T}\rfloor+2,\rho+b,b+1)$-function. 
\end{corollary}

\begin{corollary}
    The Hamming weight distribution function $\Delta_T^H$ is a locally $\left(\left\lfloor\frac{2\rho}{T}\right\rfloor+2,\rho+b-1,b\right)$-function.
\end{corollary}

In the following theorem, we determine a value of $\lambda$ such that the $b$-symbol weight distribution function $\Delta_T^b$ is a locally $(\lambda, \rho, b)$-function. The proof follows a similar approach to that of \cite[Theorem 8]{Rajput2025}. Thus, we omit the proof. 
\begin{theorem}\label{bwdf}
    The $b$-symbol weight distribution function $\Delta_T^b$ is a locally  $\left(\left\lfloor\frac{2\rho}{T}\right\rfloor+2,\rho,b\right)$-function.
\end{theorem}

\begin{corollary}
    The $b$-symbol weight function $w_b$ is a locally $(4t+2,2t,b)$-function for $t\in \mathbb{N}$.
\end{corollary}

Determining the smallest $\lambda_s$ for which the weight distribution function is a locally $(\lambda,\rho,b)$-function presents a non-trivial problem. In the subsequent results, we establish both upper and lower bounds on $\lambda_s$. 

\begin{proposition}\label{ballbdd}
   Let  $\Delta_T^b$ be a locally $(\lambda_s,\rho,b)$-function for some $0\leq \rho \leq k$ and positive integer $T$. Then  $\lambda_s\geq\left\lfloor\frac{\rho} {T}\right\rfloor + 1$ for $t\in \mathbb{N}$.
\end{proposition}
\begin{proof}
    Note that $$|B_b(\mathbf{0},\rho)|=\sum_{i=0}^\rho|S^b(\mathbf{0},i)|,$$
    where $|S^b(\mathbf{0},i)|=|\{x\in \mathbb{F}_2^k|\ d_b(x,\mathbf{0})=w_b(x)=i\}|$ for $0\leq i\leq \rho$. Also, 
    $$\Delta_T^b(S^b(\mathbf{0},i))= \left \{\left\lfloor\frac{i}{T}\right\rfloor\right\}.$$ Thus $\max_{u\in \mathbb{F}_2^k} |B^b_{\Delta_T^b}(u,\rho)|\geq \left\lfloor \frac{\rho}{T}\right\rfloor + 1 $. This completes the proof.
\end{proof}
Combining Theorem \ref{bwdf} and Proposition \ref{ballbdd}, we have the following lower and upper bounds on $\lambda_s$.
\begin{corollary}\label{luboundwd}
    Let the $b$-symbol weight distribution function $\Delta^b_T$  be a locally $(\lambda_s,\rho,b)$-function. Then   $\left\lfloor\frac{\rho}{T}\right\rfloor + 1\leq \lambda_s\leq \left\lfloor\frac{2\rho}{T}\right\rfloor+2$.
\end{corollary}

\begin{corollary}\label{luboundw}
    Let $b$-symbol weight function be a locally $(\lambda_s,\rho,b)$-function. Then $\rho+1\leq \lambda_s\leq 2\rho+2$. 
\end{corollary}
In Table \ref{tab1}, we determine the exact value of $\lambda_s$ using MAGMA that achieves the upper bound, as well as intermediate values, for various values of $k, b,\rho$, and $T$.   
\begin{table}[h!]
    \centering
    \begin{tabular}{|c|c|c|c|c|c|}
    \hline
     $k$    &$b$& $\rho$& $T$ & $\left \lfloor\frac{\rho}{T}\right\rfloor + 1\leq \lambda_s\leq \left \lfloor\frac{2\rho}{T}\right \rfloor+2$ & \text{Exact} $\lambda_s$  \\
     \hline
    $6$ & $2$ & $2$ & $3$ & $2\leq \lambda_s\leq 3$ & $3$\\
     \hline
      $6$ & $2$ & $2$ & $1$ & $3\leq \lambda_s\leq 6$ & $4$\\
     \hline
     $8$ & $2$ & $4$ & $2$ & $3\leq \lambda_s\leq 6$ & $5$\\
     \hline
      $8$ & $3$ & $4$ & $3$ & $2\leq \lambda_s\leq 4$ & $3$\\
     \hline
      $9$ & $3$ & $3$ & $7$ & $1\leq \lambda_s\leq 2$ & $2$\\
     \hline
      $10$ & $3$ & $3$ & $6$ & $1\leq \lambda_s\leq 3$ & $2$\\
     \hline
    \end{tabular}
    \caption{Locality of $b$-symbol weight distribution function over $\mathbb{F}_2^k$.}
    \label{tab1}
\end{table}

\section{Redundancy of FCBSC for locally $(\lambda,\rho,b)$-functions}\label{sectionredundancy}

Rajput et al.~\cite{Rajput2025} derived several bounds on the optimal redundancy of $(f,t)$-FCCs corresponding to locally $(\lambda,\rho)$-functions with respect to the Hamming distance. They further demonstrated that the optimal redundancy is achievable for locally $(4,2t)$-functions under certain conditions. In this section, we derive a recurrence relation connecting the optimal redundancy for $b$ and $(b+1)$-symbol read channels. Furthermore, we establish upper bounds on the optimal redundancy of FCBSCs for locally bounded functions over $b$-symbol read channels. These findings generalize several of the bounds established in \cite{Rajput2025} to the broader setting of the \(b\)-symbol distance model. 

\begin{lemma}\label{bound_lambdat}
    Let $N_H(\lambda, 2t)$ be the minimum possible length of a binary error-correcting code with $\lambda$ codewords and minimum Hamming distance $2t$. Then,  $N_H (\lambda, 2t) \leq \lambda t$.
\end{lemma}
\begin{proof}
     Let $\mathcal{C} = \{c_1,c_2,\ldots,c_\lambda\}$, where $c_i$ is a vector of length $\lambda t$ consisting of $(i-1)t \text{ } 0s$ followed by $t \text{  } 1s$ followed by $(\lambda - i)t$ $0s$. Then clearly $\mathcal{C}$ is a binary error-correcting $(\lambda t, \lambda, 2t)$ code. Therefore, $ N_H (\lambda, 2t) \leq \lambda t.$
\end{proof}
Using Theorem \ref{thm7frm11} and Lemma \ref{bound_lambdat}, we derive an upper bound on the optimal redundancy of an $(f,t)$-FCBSC for locally $(\lambda,2t,b=1)$-function $f$, as stated in the following theorem. 

    \begin{theorem}\label{le_bound_lambda}
    Let $t$ be a positive integer. For any locally $(\lambda, 2t, b=1)$-bounded function $f$, the optimal redundancy of an $(f, t)$-FCBSC is bounded above by
    \begin{equation*}
    r_f^H(k, t) \leq \lambda t.
    \end{equation*}
    \end{theorem}

Next, we show that the above bound is achievable. We provide an example demonstrating the achievability conditions.
    \begin{theorem}\label{optimality_condition}
    Let $\lambda=3$ and $t>0$ be an integer. Then, for any locally $(3,2t,b=1)$-function, the optimal redundancy of an $(f,t)$-FCBSC satisfies $r_f^H(k,t)\leq 3t$. Moreover, for a locally $(3, 2t,b=1)$-function $f$ with $|\text{Im}(f)| \geq 3$, if there exist $u_1, u_2, u_3 \in \mathbb{F}_2^k$ with $f(u_1) \neq f(u_2) \neq f(u_3)\neq f(u_1)$ such that 
    $d_H(u_1, u_2) =1, d_H(u_3, u_1)=1 \ \text{and} \ d_H(u_3, u_2)=2,$
    then $r^H_f(k, t)=3t$. 
    \end{theorem}
    \begin{proof}
    By Theorem \ref{le_bound_lambda}, we have $r_f^H(k, t) \leq 3t$. For $u_1, u_2, u_3 \in \mathbb{F}_2^k$, the distance requirement matrix is
    $$\mathcal{D}_f(t, u_1, u_2, u_3) = \left[\begin{matrix} 0 & 2t & 2t \\
    2t & 0& 2t-1 \\
    2t &2t-1 & 0  
    \end{matrix}\right].$$
    From the generalized Plotkin bound  in Lemma \ref{lemma1frm7}, we have 
    \begin{align*}
    N(\mathcal{D}_f(t, u_1, u_2, u_3)) &\geq \frac{4}{3^2-1} (D_{12}  + D_{13} + D_{23}) \\
    & = \frac{1}{2} (6t-1) = 3t-\frac{1}{2}.
    \end{align*}
    Since $N(\mathcal{D}_f(t, u_1, u_2, u_3))$ is an integer, by \cite[Corollary 1]{Lenz2023}, we have 
    $$r_f^H(k,t) \geq N(\mathcal{D}_f(t, u_1, u_2, u_3)) \geq 3t.$$ 
    Hence $r_f^H(k,t)=3t$.
    \end{proof}
    
    The following example shows that the assumptions in Theorem \ref{optimality_condition} are not absurd.
    \begin{example}
    Consider the Hamming weight function $ w_H: \mathbb{F}_2^2 \rightarrow \{0,1,2\} $. Let $ u_1 = (1,0) ,  u_2 = (0,0) , \text{ and } u_3 = (1,1) $.  Observe that the function values are
    \[
    w_H(u_1) = 1, \quad w_H(u_2) = 0, \quad w_H(u_3) = 2.
    \]
    Additionally, the Hamming distances between these vectors are
    \[
    d_H(u_1, u_2) = 1, \quad d_H(u_1, u_3) = 1, \quad d_H(u_2, u_3) = 2.
    \]
    Hence, \( w_H \) satisfies the required conditions.
    \end{example}
    
    \begin{remark}
        Theorem 4.3 shows that the bound established in Theorem 4.2 is tight for locally $(3, 2t, b = 1)$-functions under certain conditions. Moreover, the upper bound is also found to be tight for locally binary functions in \cite[Lemma 5]{Lenz2023}.
    \end{remark} 
    
    In the following theorem, we present a recurrence relation that connects the optimal redundancy of an  $(f,t)$-FCBSC for $b$-symbol and $(b+1)$-symbol read channels.  This result provides a systematic framework for estimating the optimal redundancy required as the symbol-read length increases.
    \begin{proposition}
        Let $f:\mathbb{F}_2^k\to \text{Im}(f)$. Then for any $(f,t)$-FCBSC, we have 
        $$r_f^{b+1}(k,t)\leq r_f^b(k,t),$$
        where $r_f^{b+1}(k,t)$ and  $r_f^b(k,t)$ are the optimal redundancy of an $(f,t)$-FCBSC for $(b+1)$-read  and $b$-read channels, respectively.    
    \end{proposition}
    \begin{proof}
    Let $r_f^b(k,t)=r$ be the optimal redundancy of an $(f,t)$-FCBSC. Suppose the corresponding encoding is given by $Enc(u)=(u,u_p)$, where $u_p$ is the redundancy bits in $\mathbb{F}_2^r$. For $f(u)\neq f(v)$, we have $d_b(Enc(u),Enc(v))\geq 2t+1$. Also,
    $$d_{b+1}(Enc(u),Enc(v))\geq d_b(Enc(u),Enc(v))+1\geq 2t+1.$$
    Thus $r_f^{b+1}(k,t)\leq r_f^b(k,t)$.
    \end{proof}

\begin{lemma}\cite[Lemma 1]{Rajput2025}\label{lemma1frm11}
Let $b=1$ and $f$ is a locally $(\lambda,\rho,b)$-function. Then there exists a function $\tau_f:\mathbb{F}_2^k\to \text{Im}(f)$ such that for any $u,v\in \mathbb{F}_2^k$, if $f(u)\neq f(v)$ and $d_H(u,v)\leq \rho$,  then $\tau_f(u)\neq \tau_f(v)$.   
\end{lemma}
The following lemma plays a crucial role while constructing redundancy bits. The proof follows by employing a vertex coloring approach on an appropriately defined graph, analogous to the technique utilized in Lemma \ref{lemma1frm11} for the Hamming distance setting.

\begin{lemma}\label{tau}
Let $f:\mathbb{F}_2^k\to \text{Im}(f)$ be a locally $(\lambda,\rho,b)$-function. Then there exists a map $\tau_f^b:\mathbb{F}_2^k\to [\lambda]$ such that for all $u,v\in\mathbb{F}_2^k$ with $f(u)\neq f(v)$ and $d_ b(u,v)\leq \rho$, we have $\tau_f^b(u)\neq \tau_f^b(v)$.
\end{lemma}
\begin{lemma}\label{lambda4}
 Let $t>0$ and $b\geq 1$ be integers and $f:\mathbb{F}_2^k\to \text{Im}(f)$ be a locally $(4,2t,b)$-function. Then for an $(f,t)$-FCBSC, the optimal redundancy satisfies
 $$r_f^b(k,t)\leq 3t-b+1.$$
\end{lemma}
\begin{proof}
 By Lemma \ref{tau}, there exists a map $\tau_f^b:\mathbb{F}_2^k\to [4]$ such that $\tau_f^b(u)\neq\tau_f^b(v)$ for all $u,v\in \mathbb{F}_2^k$ whenever $f(u)\neq f(v)$ and $d_b(u,v)\leq 2t$. Define 
$$u_p= \begin{cases}
000\hspace{1 cm}\text{ if } \hspace{0.5 cm}\tau_f^b(u)=1\\
110\hspace{1 cm}\text{ if } \hspace{0.5 cm}\tau_f^b(u)=2\\
101\hspace{1 cm}\text{ if } \hspace{0.5 cm}\tau_f^b(u)=3\\
011\hspace{1 cm}\text{ if } \hspace{0.5 cm}\tau_f^b(u)=4.
 \end{cases}$$
Write $(u_p)^t=(u_p',u_p'')$, where $(u_p)^t\in \mathbb{F}_2^{3t}$ is the $t$ fold repetition of $u_p$, $u_p'\in \mathbb{F}_2^{3t-b+1}$  is the first $(3t-b+1)$ coordinates of $(u_p)^t$, and $u_p''\in \mathbb{F}_2^{b-1}$ is the last $(b-1)$ coordinates of $(u_p)^t$. Define an encoding map $\mathcal{E}:\mathbb{F}_2^k\to \mathbb{F}_2^{k+(3t-b+1)}$ as $\mathcal{E}(u)=(u,u_p')$.
Observe that, for $u,v\in \mathbb{F}_2^k$, we have
\begin{align*}
    d_H(u_p^t,v_p^t)=d_H(u_p',v_p')+d_H(u_p'',v_p'').
\end{align*}
Also, $d_H(u_p^t,v_p^t)=2t$ and $d_H(u_p'',v_p'')\leq b-1$. Thus,
\begin{equation}\label{eq2t-b+1}
    d_H(u_p',v_p')=d_H(u_p^t,v_p^t)-d_H(u_p'',v_p'')\geq 2t-(b-1).
\end{equation}
Now, we show that the above encoding defines an $(f,t)$-FCBSC with redundancy $r_f^b(k,t)=3t-b+1$ for $b$-symbol read channels. Let $u,v\in \mathbb{F}_2^k$ with $f(u)\neq f(v)$. \\
\textbf{Case 1}: If $d_b(u,v)\geq 2t+1$, then 
$$d_b(\mathcal{E}(u),\mathcal{E}(v))\geq d_b(u,v)\geq 2t+1.$$
\textbf{Case 2}: If $d_b(u,v)\leq 2t$, then $\tau_f^b(u)\neq \tau_f^b(v)$. Thus, $u_p'\neq v_p'$.
If $d_H(\mathcal{E}(u),\mathcal{E}(v))> k+(3t-b+1)-(b-1)$, then by Lemma \ref{handb},
\begin{align*}
    \begin{split}
d_b(\mathcal{E}(u),\mathcal{E}(v))=k+3t-b+1\geq 2t+1\hspace{0.5 cm}  \text{ (since $k\geq b$) }.  
    \end{split}
\end{align*}
Otherwise, 
\begin{align*}
    \begin{split}
d_b(\mathcal{E}(u),\mathcal{E}(v))&\geq d_H(\mathcal{E}(u),\mathcal{E}(v))+(b-1)  \\
  &=d_H(u,v)+d_H(u_p',v_p')+(b-1)\\
  &\geq 2t+1 \hspace{0.5 cm} ( \text{by Eq.} \ref{eq2t-b+1}).
    \end{split}
\end{align*}
This completes the proof.
\end{proof}
\begin{corollary}\cite[Lemma 2]{Rajput2025}
 Let $t$ be a positive integer and $f:\mathbb{F}_2^k\to \text{Im}(f)$ be a locally $(4,2t,b=1)$-function. Then the optimal redundancy of an $(f,t)$-FCC satisfies
 $r_f^H(k,t)\leq 3t$.
\end{corollary}

Next, we extend Lemma \ref{lambda4} for an arbitrary $\lambda$. 

\begin{theorem}\label{irreb}
    Let $t>0$ be an integer and $f:\mathbb{F}_2^k\to \text{Im}(f)$ be a locally $(\lambda,2t,b)$-function. Then for any $(f,t)$-FCBSC, an upper bound on redundancy is given by 
    $$r_f^b(k,t)\leq N_b(\lambda,2t),$$
    where $N_b(\lambda,2t)$ is the minimum possible length of a binary $b$-symbol error-correcting code with $\lambda$ codewords and minimum $b$-symbol distance $2t$.
\end{theorem}
\begin{proof}
    Suppose $C$ is a binary $b$-symbol error-correcting code of length $N_b(\lambda,2t)$ with $\lambda$ codewords and minimum $b$-symbol distance $2t$. By Lemma \ref{tau}, there exists a map $\tau_f^b:\mathbb{F}_2^k\to [\lambda]$ such that for $u,v\in \mathbb{F}_2^k$, $\tau_f^b(u)\neq \tau_f^b(v)$ whenever $f(u)\neq f(v)$ with $d_b(u,v)\leq 2t$. We fix an ordering of codewords in $C$, say $C_1,C_2,\dots,C_{\lambda}$. Define an encoding map as follows
$\mathcal{E}:\mathbb{F}_2^k\to\mathbb{F}_2^{k+N_b(\lambda,2t)}$ given by $\mathcal{E}(u)=(u,C_{\tau_f^b(u)})$.
This encoding gives an $(f,t)$-FCBSC with redundancy $N_b(\lambda,2t)$. For this, let $u,v\in \mathbb{F}_2^k$ such that $f(u)\neq f(v)$. \\
\textbf{Case 1}: If $d_b(u,v)\geq 2t+1$, then 
$$d_b(\mathcal{E}(u),\mathcal{E}(v))\geq d_b(u,v)\geq 2t+1.$$\\
\textbf{Case 2}: If $d_b(u,v)\leq 2t$, then $\tau_f^b(u)\neq \tau_f^b(v)$. Thus $C_{\tau_f^b(u)}\neq C_{\tau_f^b(v)}$. By Lemma \ref{lemma3.1frm12}, 
\begin{align*}
    \begin{split}
d_b(\mathcal{E}(u),\mathcal{E}(v))\geq& d_b(u,v)+d_b(C_{\tau_f^b(u)},C_{\tau_f^b(v)})-(b-1)\\
&=b+2t-(b-1)=2t+1.  
    \end{split}
\end{align*}
This completes the proof.
\end{proof}

From Theorem \ref{irreb} and Lemma \ref{nbnh}, the following upper bound on the optimal redundancy of FCBSCs for locally $(\lambda,2t,b)$-functions can be concluded.
\begin{theorem}\label{genlambda}
    Let $t>0$, $\lambda > 2^{b-1}$ be a positive integer, and $f:\mathbb{F}_2^k\to \text{Im}(f)$ be a locally $(\lambda,2t,b)$-function. Then, for any $(f,t)$-FCBSC, an upper bound on the optimal redundancy is given by 
    $$r_f^b(k,t)\leq N_H(\lambda,2t-b+1),$$
    where $N_H(\lambda, 2t- b+1)$ is the minimum possible length of a binary error-correcting code with $\lambda$ codewords and minimum Hamming distance $2t-b+1$.
\end{theorem}

\begin{theorem}\label{2^b_2t}
    Let $t>0$ and $b\geq 1$ be integers with $b$ dividing $t$ and $N_b(\lambda, 2t)$ be the minimum possible length of a binary $b$-symbol error-correcting code with $\lambda$ codewords and minimum $b$-symbol distance $2t$. Then 
    $$N_b(2^b,2t) = 2t.$$
\end{theorem}
\begin{proof}
   Let $S=\{s_i\in \mathbb{F}_2^b|\ s_i \text{ is binary representation of } i \ \forall\  0\leq i\leq 2^b-1 \text{ of length } b\}.$ Consider binary $b$-symbol code $ \mathcal{C} = \{c_1, c_2, \ldots, c_{2^b}\}$ of size $2^b$ and length $2t$, where $c_i = (s_i)^{\frac{2t}{b}}\in \mathbb{F}_2^{2t}$. We claim that the minimum $b$-symbol distance of $C$ is at least $2t$. For this, note that $d_b(s_i,s_j)\geq b$ for $i\neq j$.  Using similar arguments as in Lemma \ref{tuplequal} (see Eq. \ref{mcopies}), for $i\neq j$, we have 
\begin{align}\label{ui2t}
    \notag d_b(c_i,c_j)=d_b(\underbrace{(s_i,s_i,\dots,s_i)}_{\frac{2t}{b} \text{ copies }},\underbrace{(s_j,s_j,\dots,s_j)}_{\frac{2t}{b} \text{ copies }})=w_b(\underbrace{s_i-s_j,s_i-s_j,\dots,s_i-s_j)}_{\frac{2t}{b} \text{ copies}}\\
   =\frac{2t}{b}w_b(s_i-s_j)=\frac{2t}{b}d_b(s_i,s_j)\geq  2t. 
\end{align}
Therefore, we get  $$N_b(2^b,2t) \leq 2t.$$
As $N_b(2^b,2t)$ is the length of a code with minimum $b$-symbol distance $2t$, thus $N_b(2^b,2t)\geq 2t$. This completes the proof.
\end{proof}
From Theorem \ref{irreb}  and Theorem \ref{2^b_2t}, we derive an upper bound on the optimal redundancy of $(f, t)$-FCBSCs for locally $(2^b, 2t, b)$-functions, as stated in the following corollary.
\begin{corollary}
Let $t>0$ and $b\geq 1$ be integers with $b \mid t$. The optimal redundancy of an $(f,t)$-FCBSC for locally $(\lambda=2^b,2t,b)$-functions  is upper bounded by 
$r_f^b(k,t)\leq 2t.$
\end{corollary}


By integrating Corollary \ref{cor3.3frm12} with the preceding theorem, we have the following corollary.

\begin{corollary}\label{opt_red_bdd}
    Let $b \geq 1$ and $t > b - 1$ be integers with $b \mid t$. The optimal redundancy of an $(f,t)$-FCBSC for any locally $(\lambda=2^b,2t,b)$-function $f$ satisfies
    $$2(t-b+1)\leq r_f^b(k,t)\leq 2t.$$
\end{corollary}

Now we present an explicit construction of an encoding function corresponding to the $b$-symbol weight distribution function $\Delta_T^b$. We show that this encoding function defines a $(\Delta_T^b,t)$-FCBSC with redundancy $2t-b+1$.

\begin{theorem}\label{explicitcons}
    Let $4t\geq T> 2t$. Then there exists a $(\Delta_T^b, t)$-FCBSC for $b$-symbol read channel with redundancy $r_{\Delta_T^b}^b(k,t)=2t-b+1$. 
\end{theorem}
\begin{proof}
Define an encoding $\mathcal{E}:\mathbb{F}_2^k\to \mathbb{F}_2^{k+2t-b+1}$ as follows
$$\mathcal{E}(u)=(u,u_p'),$$ where 
$u_p'=(u_p)^{2t-b+1}\in \mathbb{F}_2^{2t-b+1}$ and 
$$u_p=\begin{cases}
    0 \hspace{0.5 cm} \text{if } \Delta_T^b(u)\equiv 0\pmod{2}\\
    1 \hspace{0.5 cm} \text{if } \Delta_T^b(u)\equiv 1\pmod{2}.\\
\end{cases}$$
This encoding gives a $(\Delta_T^b,t)$-FCBSC. For this, let $u,v\in \mathbb{F}_2^k$ such that $\Delta_T^b(u)\neq \Delta_T^b(v)$. Then, we have the following.\\
\textbf{Case 1}: If $\Delta_T^b(u)$ and $\Delta_T^b(v)$ are both odd or both even, then (WLOG assume that $\Delta_T^b(u)>\Delta_T^b(v)$)
$$\Delta_T^b(u)-\Delta_T^b(v)\geq 2.$$
Also, 
$$\frac{w_b(u)-w_b(v)}{T}\geq \left \lfloor\frac{w_b(u)}{T}\right \rfloor-\left \lfloor \frac{w_b(v)}{T}\right \rfloor-1=\Delta_T^b(u)-\Delta_T^b(v)-1\geq 1.$$
Thus, $d_b(u,v)\geq w_b(u)-w_b(v)\geq T\geq 2t+1$. Consequently, $$d_b(\mathcal{E}(u),\mathcal{E}(v))\geq d_b(u,v)\geq 2t+1.$$
\textbf{Case 2}: If $\Delta_T^b(u)$ is odd and $\Delta_T^b(v)$ is even (or $\Delta_T^b(u)$ is even and $\Delta_T^b(v)$ is odd), then $d_H(u_p',v_p')=2t-b+1$.
Now, if $d_H(\mathcal{E}(u),\mathcal{E}(v))\leq (k+2t-b+1)-(b-1)$, then by Lemma \ref{handb}, we have
\begin{align*}
    \begin{split}
        d_b(\mathcal{E}(u),\mathcal{E}(v))&\geq d_H(\mathcal{E}(u),\mathcal{E}(v))+(b-1)\\
        &=d_H(u,v)+ d_H(u_p',v_p')+(b-1)\\
        &\geq 1+(2t-b+1)+(b-1)= 2t+1.
    \end{split}
\end{align*}
Otherwise, when $d_H(\mathcal{E}(u),\mathcal{E}(v))> (k+2t-b+1)-(b-1)$. Then
\begin{align*}
    \begin{split}
        d_b(\mathcal{E}(u),\mathcal{E}(v))=k+2t-b+1\geq 2t+1 \hspace{0.2 cm}\text{ (since } k\geq b).
    \end{split}
\end{align*}
This completes the proof.
\end{proof}

\section{Conclusion}\label{sectionconclusion}
    In this work, we considered function-correcting codes designed for locally bounded functions over $b$-symbol read channels, called locally $(\lambda,\rho,b)$-functions. We investigated the values of $\lambda $ and $\rho$ for which a given function is locally bounded by $\lambda$ in the $b$-symbol metric. Further, we explicitly explored the smallest values of $\lambda$ for $b$-symbol weight and weight distribution functions to be locally bounded functions. We associated the optimal redundancy of function-correcting $b$-symbol codes (FCBSCs) with the smallest achievable length of error-correcting codes in $b$-symbol metric. Using this, we established several upper bounds on the optimal redundancy of FCBSCs for locally $(\lambda,\rho,b)$-functions. We showed that, under certain conditions, a locally ($3,2t,b=1$)-function achieves the optimal redundancy of $3t$. Furthermore, we explicitly studied the optimal redundancy of FCBSCs for the weight functions and the weight distribution functions in the Hamming and the $b$-symbol metric. In addition, we provided explicit construction of FCBSCs for $b$-symbol weight distribution function with a redundancy $2t-b+1$.
    
    The study of function-correcting codes for different classes of functions under various metrics is an interesting and active area of research. Therefore, extending the study of function-correcting codes to other function classes, such as linear functions and diverse read channels, remains an important avenue for future investigation. In particular, the study of function-correcting codes for the class of linear functions and locally bounded functions under the homogeneous metric is an interesting problem.

\bibliographystyle{abbrv}
	\bibliography{ref}

\end{document}